\documentclass[twocolumn]{IEEEtran}

\usepackage{amsfonts,amssymb,amsmath,amsthm,xcolor}
\usepackage{bm,url}
\newtheorem{example}{Example}
\newtheorem{theorem}{Theorem}
\newtheorem{definition}{Definition}
\usepackage{caption} 
\usepackage{cite}
\usepackage{graphicx}
\newcommand{\EE}[1]{{\rm{E}}\left[#1\right]}
\usepackage{hyperref}

\title{On the normalized signal to noise ratio in covariance estimation}
\author{Tzvi Diskin and Ami Wiesel}

\begin{document}
\maketitle

\begin{abstract}
    We address the Normalized Signal to Noise Ratio (NSNR) metric defined in the seminal paper by Reed, Mallett and Brennan on adaptive detection. The setting is detection of a target vector in additive correlated noise. NSNR is the ratio between the SNR of a linear detector which uses an estimated noise covariance and the SNR of clairvoyant detector based on the exact unknown covariance. It is not obvious how to evaluate NSNR since it is a function of the target vector. To close this gap, we consider the NSNR associated with the worst target. Using the Kantorovich Inequality, we provide a closed form solution for the worst case NSNR. Then, we prove that the classical Gaussian Kullback Leibler (KL) divergence bounds it. Numerical experiments with different true covariances and various estimates also suggest that the KL metric is more correlated with the NSNR metric than competing norm based metrics.
\end{abstract}

\section{Introduction}

Covariance estimation is a fundamental problem in statistical signal processing \cite{kay1998fundamentals,wiesel2015structured,li1999computationally,aubry2017geometric,mestre2008asymptotic,smith2005covariance}. A large body of works considers different algorithms in various structures, distributions and settings. This letter addresses the natural question: Given an estimate $\hat C$ of a true covariance $C$, what is the best way to measure its accuracy?

There are numerous competing distance functions between $C$ and $\hat C$. These include algebraic distances as the Frobenius norm or the spectral norm, as well as statistical divergences as the Kullback Leibler (KL) divergence \cite{kullback1951information}. The metrics differ in their statistical properties and computational complexity. Choosing between them is challenging because covariance estimation is typically an intermediate step within the signal processing pipeline. The estimate is less interesting on its own and is mainly used as an input to a downstream task as target detection, beamforming or Portfolio design \cite{zhang2013finite}. Therefore, accuracy should be measured with respect to the end goal rather than the covariance itself. For example, the end goals in the detection are the error probabilities. One can simulate the complete end to end pipeline and define accuracy using these probabilities. But this approach is not modular, depends on the specific detection scenario, has no closed form solution and is computationally expensive. Therefore, it is more common to measure the distance using the simpler surrogate metrics. 

The choice of metric regains attention over the last years due to increasing use of machine learning. Originally, metrics were used for performance analysis and comparison of existing algorithms \cite{bickel,eldar2020sample}. In machine learning, metrics are also used as loss functions and are an important decision in the design and training of neural networks. In the context of covariance estimation, there are many recent works on covariance prediction \cite{barratt2023covariance,barthelme2021doa,addabbo2022nn,kang2023clutter,diskin2024self}. The most common loss functions are the Frobenius norm and KL divergence. It is unclear how these loss functions correlate with the end to end detection performance.

We focus on the use of covariance estimation for detecting a target in additive correlated noise using an adaptive matched filter \cite{reed1974rapid,kelly1986adaptive,krim1996two,kay1998fundamentals,manolakis2013detection}. Signal to Noise Ratio (SNR) is the key performance measure in this application as it controls the probabilities of error. The SNR decreases when an inexact covariance is used. Therefore, in their seminal paper, Reed, Mallett and Brennan (RMB) proposed to measure performance using the Normalized Signal to Noise Ratio (NSNR)  \cite{reed1974rapid}:
\begin{eqnarray}\label{nsnr}
    &{\rm{NSNR}}& = \frac{{\rm{SNR}}_{{\rm{estimated\;}}\hat C}}{{\rm{SNR}}_{{\rm{true\;}}C}}
\end{eqnarray}  
Intuitively, as $\hat C \rightarrow C$ we get NSNR $\rightarrow 1$ and NSNR is indicative of the distance between $C$ and $\hat C$. RMB proved the following rule of thumb for a target of length $D$ \cite{reed1974rapid}: ``If one wishes to maintain an average loss ratio of better
than one-half then at least $2D$ samples of data are needed''. This rule is based on ideal probabilistic assumptions and holds only for the sample covariance estimate. 
If the true signal direction is not known exactly then more samples are required to ensure near optimal performance \cite{boroson1980sample}. Numerical analysis of NSNR using regularized sample estimates was presented in \cite{carlson1988covariance}. A main difficulty in SNR based measures is that these depend on the target vector and are only applicable if it is a priori known. Some works bypass this by assuming that the targets are random \cite{mestre2005finite,besson2005performance}. To our knowledge, an analysis of the NSNR with respect to arbitrary targets in general settings and more sophisticated estimation methods is still missing.

The main contribution of this work is the introduction and analysis of the worst case NSNR distance with respect to the target vector:
\begin{align}
    d_{\rm{NSNR}}(C,\hat C) = \max_{\rm{target\;vector}} -\frac{1}{2}\log\;{\rm{NSNR}}
\end{align}
By changing variables and applying the Kantorovich inequality, we provide a closed form solution to the distance. It is non-differentiable and involves eigenvalues. But we prove that it is upper bounded by the popular KL divergence. Therefore, we advocate for using KL when analyzing or optimizing the NSNR.

We conclude the paper with simple experiments that numerically investigate the correlation between different metrics and the NSNR distance. The results clearly support the claim that KL based metrics are more correlated with NSNR than the norm based metrics. This is consistent across various settings with different number of samples and different regularization parameters of the sample covariance. Interestingly, previous work used the Frobenius norm to choose the best regularization  \cite{ledoit2004well,chen2010shrinkage,chen2012shrinkage,ollila2022regularized}. Future work on covariance estimation for target detection may consider tuning the parameter using KL which is apparantly more correlated with the end performance.

\section{Normalized Signal to Noise Ratio}
We begin by recalling the definition of NSNR from \cite{reed1974rapid}. We consider the detection of a known signal in additive noise \cite{reed1974rapid,kelly1986adaptive,krim1996two,kay1998fundamentals,manolakis2013detection}. The observed vector is 
\begin{align}
    x = As+n
\end{align}
where $s$ is a length $D$ signal vector and $n$ is a zero mean noise vector with a positive definite covariance matrix $C\succ 0$. The goal is to decide between
\begin{align}
    & H_0:\quad A=0\nonumber\\
    & H_1:\quad A>0
\end{align}
The classical test for this setting is 
\begin{align}
    w^Tx \lessgtr \gamma
\end{align} 
where $\gamma$ is a threshold parameter and $w$ is a (decorrelated) matched filter
\begin{align}
    w = C^{-1}s.
\end{align}
Assuming Gaussian noise, this detector is also the Likelihood Ratio Test (LRT). 

The test is linear and its performance is  characterized by the resulting SNR:
\begin{eqnarray}
    &{\rm{SNR}}_{{\rm{true\;}}C}& = \frac{(w^Ts)^2}{E[(w^Tn)^2]} \nonumber\\
    &&= \frac{(s^TC^{-1}s)^2}{s^TC^{-1}E[nn^T]C^{-1}s}\nonumber\\
    &&=\frac{(s^TC^{-1}s)^2}{s^TC^{-1}CC^{-1}s}\nonumber\\
    &&=s^TC^{-1}s
\end{eqnarray}
In practice, the true covariance matrix of the noise $C$ is rarely known. It is common to replace it with an estimate $\hat C\succ 0$ and use 
\begin{align}
    \hat w = \hat C^{-1}s
\end{align}
This leads to a mismatched and lower SNR
\begin{align}
    &{\rm{SNR}}_{{\rm{estimated\;}}\hat C} = \frac{(s^T\hat C ^{-1}s)^2}{s^T\hat C^{-1}C\hat C^{-1}s}
\end{align}
To analyze the loss in performance due to the inaccurate covariance, RMB proposed the NSNR \cite{reed1974rapid}:
\begin{definition}
Let $C$ and $\hat C$ be two positive definite matrices. The NSNR in (\ref{nsnr}) is given by
\begin{eqnarray}
    &{\rm{NSNR}}&=\frac{(s^T\hat C ^{-1}s)^2}{(s^T\hat C^{-1}C\hat C^{-1}s)(s^TC^{-1}s)}.
\end{eqnarray}    
\end{definition}

It is easy to see that
\begin{align}
    0 \leq {\rm{NSNR}}(x) \leq 1
\end{align}
and the goal of different covariance estimation methods is to try and achieve the upper bound.

Note that the NSNR depends on $s$ and the degradation in NSNR in different settings is still poorly understood. The goal of this letter is to shed more light on this important metric.

\section{Worst case NSNR}
In this section, we eliminate the dependence of NSNR on the target vector by considering the worst case scenario. It turns out that the analysis depends $C$ and $\hat C$ via their matrix ratio $Q$ which is the a key component in many covariance estimation metrics (e.g., 
\cite{tiomoko2019random}).
\begin{definition}
   Let $C$ and $\hat C$ be two positive definite matrices. Their matrix ratio $Q$ is defined as
\begin{align}
    Q = \hat C^{-\frac{1}{2}}C \hat C^{-\frac{1}{2}}.
\end{align} 
\end{definition}
Using the ratio, the next result provides a closed form solution to the worst case NSNR which is a function only of its condition number $\kappa(Q)$.
\begin{theorem}
Let $C$ and $\hat C$ be two positive definite matrices with a ratio $Q$. Then, the worst case NSNR is given by

\begin{eqnarray}\label{nsnrmin}
  &{\rm{NSNR}}_{\min}(C,\hat C)  &= \min_s \frac{(s^T\hat C ^{-1}s)^2}{(s^T\hat C^{-1}C\hat C^{-1}s)(s^TC^{-1}s)}\nonumber\\
  && = \frac {4\kappa(Q)}{\left(\kappa(Q)+1\right)^2}
\end{eqnarray}

It is symmetric
    \begin{align}
        {\rm{NSNR}}_{\min}(C,\hat C) = {\rm{NSNR}}_{\min}(\hat C, C)
    \end{align}
and bounded from above 
    \begin{align}
        {\rm{NSNR}}_{\min}(C,\hat C)\leq 1
    \end{align}
with equality if and only if $C=\alpha \hat C$ for some $\alpha>0$. 
\end{theorem}

\begin{proof}    
To simplify the optimization, we change variables
\begin{align}
    y=\hat C^{-\frac{1}{2}}s
\end{align}
and obtain
\begin{eqnarray}
  &{\rm{NSNR}}_{\min}(C,\hat C)  &= \min_y \frac{(y^Ty)^2}{(y^T\hat C^{-\frac{1}{2}}C\hat C^{-\frac{1}{2}}y)(y^T\hat C^{\frac{1}{2}}C^{-1}\hat C^{\frac{1}{2}}y)}\nonumber\\
  && = \min_y \frac{(y^Ty)^2}{(y^T\hat Qy)(y^TQ^{-1}y)}
\end{eqnarray} 
The NSNR is invariant to scaling of $y$ by a positive scalar and we can restrict the attention to the case of $y^Ty=1$:
\begin{align}\label{kattarovich optimization}
    {\rm{NSNR}}_{\min}^{-1} = 
    \left\{\begin{array}{ll}
        \max_y & (y^TQy)(y^TQ^{-1}y) \\
        {\rm{s.t.}} & y^Ty=1
    \end{array}\right.
\end{align}

The solution of the to the optimization problem is known as the Kantrovitch inequlity which states that \cite{liu1999survey}: 
\begin{equation}
    (y^TQy)(y^TQ^{-1}y) \leq \frac {\left(q_{\min}+q_{\max}\right)^2}{4q_{\min}q_{\max}}
\end{equation}
where $q_{\min}$ and $q_{\max}$ are the minimal and maximal eigenvalues of $Q$
A sufficient condition for equality is when:
\begin{equation}
    s = \frac{1}{\sqrt{2}}(u_{\min} + u_{\max}),
\end{equation}
where $u_{\min}$ and $u_{\max}$ are the eigenvectors that correspond to $q_{\min}$ and $q_{\max}$.
Finally, dividing the numerator and denominator by $q_{\min}$ gives \ref{nsnrmin}.

The upper bound ${\rm{NSNR}}_{\min}(C,\hat C)\leq 1$ is due to the arithmetic mean and geometric mean inequality.
Finally, the symmetry of the worst case NSNR is straight forward by noting that the eigenvalues of $\hat C^{-\frac 1 2}C\hat C^{-\frac 1 2}$ and $C^{\frac 1 2}\hat C^{-1}C^{\frac 1 2}$ are identical. 
\end{proof}

Given its properties, it is natural to consider the NSNR$_{\min}$ as a distance function between $C$ and $\hat C$. High NSNR indicates that $C$ and $\hat C$ are close and therefore the metric should be an inversely proportional to the NSNR. Following the convention of measuring power ratios in decibels (dB), we also take a logarithm and define the NSNR distance as
\begin{eqnarray}
    d_{\rm{NSNR}}(C,\hat C) = -\frac{1}{2}\log{\rm{NSNR}}_{\min}(C,\hat C).
\end{eqnarray}
As expected, $d_{\rm{NSNR}}\geq 0$. It is equal to zero if and only if $C=\alpha\hat C$ for some $\alpha>0$. This is not common to distance functions but is natural to the detection formulation which is invariant to scaling of the covariance. Unfortunately, it is easy to find examples that show that it does not satisfy the triangle inequality.


\section{Relation to other distances}
In this section we compare the NSNR distance with other metrics. A popular distance is the Gaussian KL divergence \cite{kullback1951information}:
\begin{eqnarray}
   & d_{\rm{KL}}(C,\hat C) &= \frac 12 {\rm{Tr}}\left(C\hat C^{-1}\right)-\frac D2-\frac 12\log \left|C\hat C^{-1}\right|\nonumber\\
   &&= \frac 12{\rm{Tr}}\left(Q\right)-\frac D2-\frac 12\log \left|Q\right|
\end{eqnarray}
where we have used the cyclic properties of the trace and determinant operators. The next theorem shows that KL is in fact a sharp upper bound to the NSNR metric.
\begin{theorem}
    Let $C$ and $\hat C$ be two positive definite matrices. Then
    \begin{equation}\label{nsnr_kl}
     d_{\rm{NSNR}}(C,\hat C)\leq d_{\rm{KL}}(C,\hat C).
    \end{equation}
\end{theorem}
\begin{proof}
\begin{align}
2d_{\rm{NSNR}} &=  \log{\rm{NSNR}}_{\min}^{-1} \nonumber\\
&= \log\frac{\left(\kappa(Q)+1\right)^2}{4\kappa(Q)}\nonumber \\
& =\log\frac {\left(q_{\min}+q_{\max}\right)^2
   }{4q_{\min}q_{\max}}\nonumber\\
& =\log\left(q_{\min}+q_{\max}\right)^2-\log \left(q_{\min}q_{\max}\right)-\log(4)\nonumber\\
& \leq q_{\min}+q_{\max}+\log(4)-2-\log \left(q_{\min}q_{\max}\right)-\log(4)\nonumber\\& = q_{\min}+q_{\max}-2-\log \left(q_{\min}q_{\max}\right)\nonumber\\
& \leq \sum_iq_i-D-\log \prod_iq_i\nonumber\\
& = 2d_{\rm{KL}}
\end{align}
where we have used
    \begin{align}
    \log(q^2)\leq q+\log(4)-2 \quad \forall\quad q>0
\end{align}
in the first inequality which holds with equality if and only if $q_{\min}+q_{\max}=2$. The second inequality is based on
\begin{align}
    q-1-\log(q)\geq 0  \quad \forall\quad q>0
\end{align}
and holds with inequality if and only if $q_i=1$ for all $i$ except $q_{\min}$ and $q_{\max}$.
\end{proof}
The relation in (\ref{nsnr_kl}) can also be explained by physical considerations. The KL divergence measures the loss of information in a log scale, while the SNR depends quadratically on the amplitude of the signal that carries the information.

Many works use other distances based on norms of the difference between the matrices. These include the Frobenius norm and spectral \cite{bickel,eldar2020sample} which cannot be expressed as functions of $Q$ and $Q^{-1}$.
\begin{eqnarray}
    & d_{\rm{Frobenius}}(C,\hat C) &= \|C-\hat C\|_{\rm{Frobenius}}\nonumber\\
    & d_{\rm{Spectral}}(C,\hat C) &= \|C-\hat C\|_2.
\end{eqnarray}
In contrast to Theorem 2, the next counter example demonstrated that small values under these norms do not necessarily lead to good NSNR.
 To avoid the trivial fact that the norm based matrices are sensitive to scaling, the example is based on matrices with near unit trace.
 \begin{example} 
Consider the following matrices:
\begin{equation}
C=\left(\begin{array}{ccc}
1 & 0 & 0\\
0 & \alpha & 0\\
0 & 0 & \alpha^2
\end{array}\right),\hat C=\left(\begin{array}{ccc}
1 & 0 & 0\\
0 & \alpha^2 & 0\\
0 & 0 & \alpha
\end{array}\right)
\end{equation}
where $0<\alpha \ll 1$.
The Frobenius and the spectral norm are given by:
\begin{align}
    d_{\rm{Frobenius}}(C,\hat C) = \sqrt 2(\alpha - \alpha^2), \qquad
    d_{\rm{Spectral}}(C,\hat C) = \left(\alpha - \alpha^2\right)
\end{align}
and are very small. On the other hand,
\begin{equation}
     d_{\rm{NSNR}}(C,\hat C) = \log \frac{1+\alpha^2}{2\alpha} 
\end{equation}
 which can be arbitrarily large as $\alpha$ approaches zero. Note that even in relatively small condition number of 10, $d_{\rm NSNR}\approx 1.6$  which means more than 95\% drop in the SNR for worst case target. 
\end{example}




\section{Numerical experiments}

\begin{table}[]
    \centering
\begin{tabular}{|l|cccc|}
\hline
& N=50 & N=100 & N = 150 & N=200 \\
\hline
Frobenius & 0.6 & 0.72 & 0.74 & 0.77 \\
Spectral Norm & 0.49 & 0.59 & 0.61 & 0.66 \\
KL & 0.81 & 0.79 & 0.84 & 0.82 \\
symKL & 0.85 & 0.83 & 0.86 & 0.85 \\
\hline
\end{tabular}
    \caption{Pearson correlation of different metrics with the worst case NSNR. The true covariance is the identity.}
    \label{tab:true_cov_is_eye_pearson}
\end{table}


\begin{table}[]
    \centering
\begin{tabular}{|l|cccc|}
\hline
& N=50 & N=100 & N = 150 & N=200 \\
\hline
Frobenius & 0.1 & 0.08 & 0.12 & 0.17 \\
Spectral Norm & 0.09 & 0.07 & 0.11 & 0.17 \\
KL & 0.82 & 0.84 & 0.81 & 0.83 \\
symKL & 0.85 & 0.87 & 0.84 & 0.85 \\
\hline
\end{tabular}
    \caption{Pearson correlation of different metrics with the worst case NSNR. The true covariance is approximately low rank.}
    \label{tab:true_cov_is_lowrankpearson}
\end{table}

\begin{table}[]
    \centering
\begin{tabular}{|l|ccc|}
\hline
& $\lambda=0.01$ &$\lambda=0.1$ & LW \\
\hline
Frobenius & 0.1 & 0.25 & 0.2 \\
Spectral Norm & 0.09 & 0.24 & 0.2 \\
KL & 0.83 & 0.85 & 0.76 \\
symKL & 0.85 & 0.81 & 0.68 \\
\hline
\end{tabular}
    \caption{Pearson correlation of different metrics with the worst case NSNR. The true covariance is approximately low rank.}
    \label{tab:reg}
\end{table}

\begin{tabular}{lccc}
\end{tabular}

\begin{figure}
    \centering
    \includegraphics[width=0.85  \linewidth]{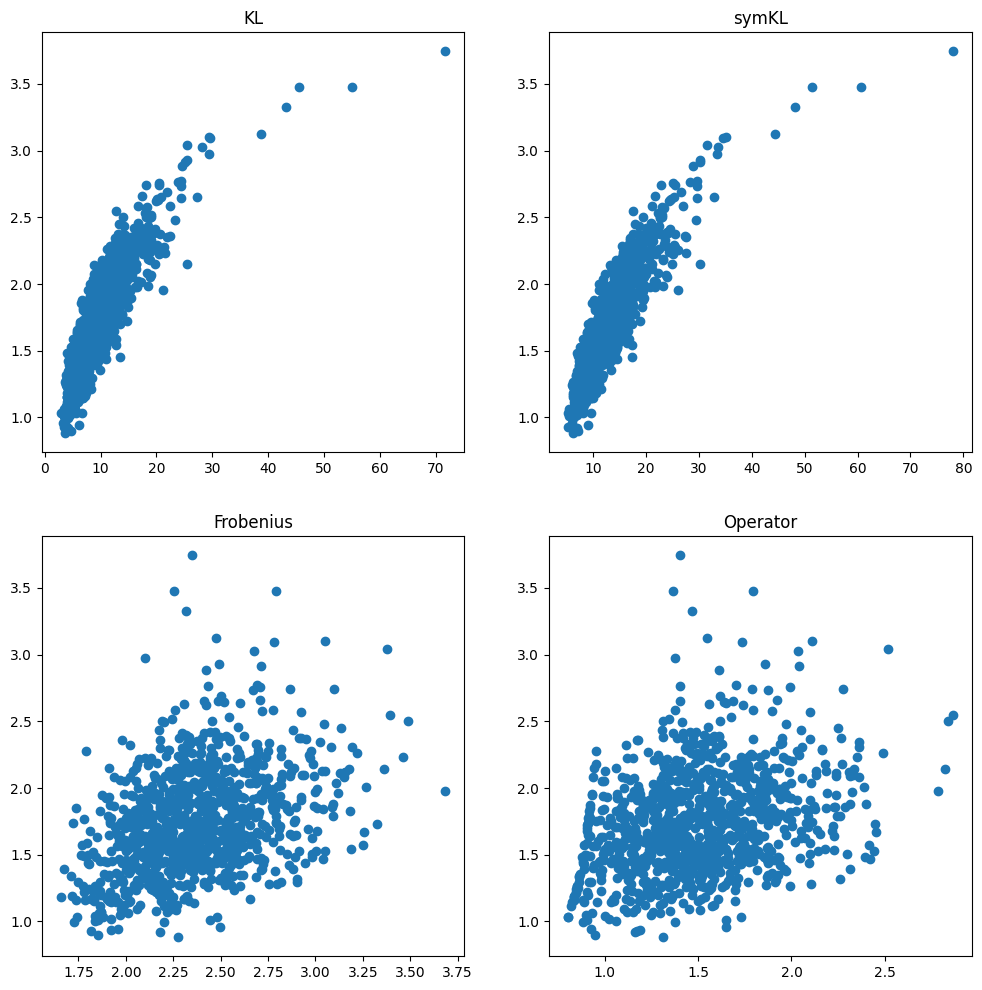}
    \caption{Scatter plots of the metrics where the y-axis are the NSNR values and the x-axis are the four metrics. It is easy to see that the KL metrics are more correlated with NSNR than the norm based metrics.}
    \label{fig:scatter}
\end{figure}

In this section, we consider the relations between NSNR and the competing metrics using numerical experiments. We emphasize that the goal is not to compare state of the art algorithms but understand the metrics better. Therefore, we use very simple settings \footnote{The code for reproducing the experiments can be found in \url{https://github.com/tzvid/nsnr}}. 

We test two ground truth covariance matrices:
\begin{itemize}
    \item Identity with $C_{\rm{identity}}=I$.
    \item Approximately low rank  with $C_{\rm{low\;rank}}=I+A$ where $A$ is all zeros except for $A_{11}=100$.
\end{itemize}
The dimension of the covariance matrix in all experiments is $m=10$.
The estimated covariance is a regularized sample covariance
\begin{align}
    \hat C_{\lambda} = \frac{1}{N}\sum_{i=1}^Nv_iv_i^T+\lambda I
\end{align}
where $\lambda>0$ is a regularization parameter and $v_i$ are $N$ independent and identically distributed realizations of a zero mean multivariate normal with the underlying covariance. We experiment with fixed $\lambda=0.01$, $\lambda=0.1$ as well as an adaptive $\lambda$ choice due to Ledoit and Wolf (LW) \cite{ledoit2004well}. 

We compare five metrics: NSNR, Frobenius norm, spectral norm, KL and as well as the symmetric version of KL denoted by symKL \cite{pereira2024asymptotics}. We measure the relation between each of the metrics with NSNR by running $1000$ independent experiments and reporting the resulting Pearson's correlation.

In the first experiment, we test the correlations under different number of samples. We use a non-regularized sample covariance with $\lambda=0$ and $N=50,\cdots,200$. Table I reports the Pearson correlation using an identity true covariance. The table clearly shows that KL and symKL are significantly more correlated with NSNR than the two other metrics. The KL based metrics are also quite stable as a function of $N$. The norm based metrics are less correlated but improve when $N$ is large. 
As visualization, Fig. 1 provides scatter plots of the metrics where the y-axis are the NSNR values and the x-axis are the four metrics. Here too it is easy to see that the KL based metrics are more correlated with the NSNR than the norms.

The second experiment is identical to the first except that the true covariance is approximately low rank. The results presented in Table III agree with previous results but are even more significant.

The third experiment considers the behavior with respect to different estimates. Again we assume an approximately low rank ground truth and compare the metrics with $N=200$ and different regularizations. The results in Table IV are consistent with the previous conclusions and demonstrate the strong correlation between NSNR and the KL based metrics.

Finally, we consider the use of the metrics for parameter tuning in covariance estimator selection. The true covariances are generated as $C=C_0+\Delta C$ where $C_0$  is low rank and $\Delta C$ is a Wishart random matrix with 20 degrees of freedom. The estimators are defined as \cite{bandiera2010knowledge}:
\begin{equation}
    \hat C_{\lambda} = (1-\lambda)\frac{1}{N}\sum_{i=1}^Nv_iv_i^T+\lambda \EE{C}.
\end{equation}
with $N=50$. The choice of $\lambda$ is done by minimizing the different metrics over $1000$ random examples. Table IV shows the chosen shrinkage coefficients $\lambda^*$ and their corresponding worst case NSNR. It can be seen that KL and symKL lead to near optimal regularization, while the other metrics result in suboptimal values and a degradation in detection performance.

\begin{table}[]
    \centering
\begin{tabular}{|l|cc|}
\hline
& $\lambda^*$ & ${\rm{NSNR}}_{\min}$  \\
\hline
Frobenius & 0.38 & 0.43  \\
Spectral Norm & 0.44 & 0.41  \\
KL & 0.04 & 0.75  \\
symKL & \textbf{0.02} & \textbf{0.76} \\
\hline
NSNR  & 0.02 & 0.76\\
\hline
\end{tabular}
    \caption{The estimated value of the shrinkage coefficient $\lambda$ and the corresponding worse case NSNR, as selected by using different metrics}
    \label{tab:reg}
\end{table}

\section{ACKNOWLEDGMENT} 
The authors would like to thank Daniel Busbib for fruitful discussions and helpful insights. Part of this research was supported by ISF grant 2672/21.

\bibliographystyle{IEEEtran}
\bibliography{main.bib}

\end{document}